\documentclass[letterpaper, 10 pt, conference]{ieeeconf}  % Comment this line out
\usepackage{graphicx}
\usepackage{amsmath}
\usepackage{mathtools}
\usepackage{amssymb}
\usepackage{amsfonts}
\usepackage{upgreek}
\usepackage[hidelinks]{hyperref}
\usepackage{cite}
\usepackage{algorithm}
\usepackage{algorithmic}
\usepackage{mathdots}
\usepackage{bm}
\usepackage{textcomp}
\usepackage{cases}
\usepackage[font=footnotesize]{caption}
\usepackage[font=footnotesize]{subcaption}
\usepackage{color}
\usepackage{soul}
\usepackage[normalem]{ulem}
\usepackage{graphbox}
\usepackage{wasysym}
\usepackage{xcolor}

\newtheorem{lemma}{Lemma}
\newtheorem{remark}{Remark}
\newtheorem{theorem}{Theorem}
\newtheorem{assumption}{Assumption}
\newtheorem{property}{Property}

\newcommand{\R}{\mathbb{R}}
\newcommand{\rsl}{\underline{\rho}^s}
\newcommand{\rsu}{\bar{\rho}^s}
\newcommand{\rhl}{\underline{\rho}^h}
\newcommand{\rhu}{\bar{\rho}^h}
\newcommand{\drsl}{\dot{\underline{\rho}}^s}
\newcommand{\drsu}{\dot{\bar{\rho}}^s}
\newcommand{\drhl}{\dot{\underline{\rho}}^h}
\newcommand{\drhu}{\dot{\bar{\rho}}^h}
\newcommand{\ru}{\rho^U}
\newcommand{\rl}{\rho^L}
\newcommand{\rU}{\rho^U}
\newcommand{\rL}{\rho^L}
\newcommand{\pl}{\varphi^L}
\newcommand{\pu}{\varphi^U}
\newcommand{\dpl}{\dot{\varphi}^L}
\newcommand{\dpu}{\dot{\varphi}^U}
\newcommand{\drl}{\dot{\rho}^L}
\newcommand{\dru}{\dot{\rho}^U}
\newcommand{\sign}{\mathrm{sign}}
\newcommand{\etl}{\eta^L}
\newcommand{\etu}{\eta^U}
\newcommand{\detl}{\dot{\eta}^L}

\newcommand{\col}{\mathrm{col}}
\newcommand{\diag}{\mathrm{diag}}
\newcommand{\xh}{\hat{x}}
\newcommand{\evh}{\hat{e}_v}
\newcommand{\evih}{\hat{e}_{v_i}}
\newcommand{\Linf}{\mathcal{L}_{\infty}}
\newcommand{\taum}{\tau_{\mathrm{max}}}
\newcommand{\gamv}{\gamma^v}
\newcommand{\epix}{\varepsilon_x}
\newcommand{\epiv}{\varepsilon_v}
\newcommand{\epixi}{\varepsilon_{x_i}}
\newcommand{\epivi}{\varepsilon_{v_i}}
\newcommand{\epixbar}{\bar{\varepsilon}_x}
\newcommand{\epivbar}{\bar{\varepsilon}_v}

\def\@IEEEtablestring{table}

\long\def\@makecaption#1#2{%
	\ifx\@captype\@IEEEtablestring%
	\begin{center}{\footnotesize #1}\\{\footnotesize\scshape #2}\end{center}%
	\@IEEEtablecaptionsepspace% V1.6 was a hard coded 8pt
	\else
	\@IEEEfigurecaptionsepspace% V1.6 was a hard coded 5pt
	\setbox\@tempboxa\hbox{\footnotesize #1.~~ #2}%
	\ifdim \wd\@tempboxa >\hsize%
	\setbox\@tempboxa\hbox{\footnotesize #1.~~ }%
	\parbox[t]{\hsize}{\footnotesize \noindent\unhbox\@tempboxa#2}%
	\else%
	\ifcenterfigcaptions \hbox to\hsize{\footnotesize\hfil\box\@tempboxa\hfil}%
	\else \hbox to\hsize{\footnotesize\box\@tempboxa\hfil}%
	\fi\fi\fi}

\IEEEoverridecommandlockouts                              % This command is only
\title{\LARGE \bf
	Funnel Control Under Hard and Soft Output Constraints \\(extended version)
}

\author{Farhad Mehdifar, Charalampos P. Bechlioulis and Dimos V. Dimarogonas% <-this % stops a space
\thanks{This work is supported by ERC CoG LEAFHOUND, H2020-ICT project CANOPIES, the KAW foundation, and the Swedish Research Council (VR).}% <-this % stops a space
\thanks{F. Mehdifar and D. V. Dimarogonas are with the Division of Decision and Control Systems, KTH Royal Institute of Technology, Stockholm, Sweden.   {\tt\small mehdifar@kth.se; dimos@kth.se}}%
\thanks{C. P. Bechlioulis is with the Division of Systems and Control of the Department of Electrical and Computer Engineering at University of Patras, Patra, Greece. {\tt\small chmpechl@upatras.gr}}%
}

\begin{document}

\maketitle
\thispagestyle{empty}
\pagestyle{empty}

%%%%%%%%%%%%%%%%%%%%%%%%%%%%%%%%%%%%%%%%%%%%%%%%%%%%%%%%%%%%%%%%%%%%%%%%%%%%%%%%
\begin{abstract}
This paper proposes a funnel control method under time-varying hard and soft output constraints. First, an online funnel planning scheme is designed that generates a constraint consistent funnel, which always respects hard (safety) constraints, and soft (performance) constraints are met only when they are not conflicting with the hard constraints. Next, the prescribed performance control method is employed for designing a robust low-complexity funnel-based controller for uncertain nonlinear Euler-Lagrangian systems such that the outputs always remain within the planned constraint consistent funnels. Finally, the results are verified with a simulation example of a mobile robot tracking a moving object while staying in a box-constrained safe space.
\end{abstract}

%%%%%%%%%%%%%%%%%%%%%%%%%%%%%%%%%%%%%%%%%%%%%%%%%%%%%%%%%%%%%%%%%%%%%%%%%%%%%%%%
\section{Introduction}

During the past decades, reference/trajectory tracking, as well as stabilization of complex and uncertain nonlinear dynamical systems, has attracted considerable research effort. Constraints are ubiquitous in controller design of practical nonlinear systems and they mainly emerge as performance and safety specifications. Constraint violation may result in performance degradation, system damage and hazards, therefore, owing to practical needs and theoretical challenges, the rigorous handling of constraints in the control design of nonlinear systems has become a dominant research topic during the past decade. Common existing methods in dealing with different types of constraints include model predictive control\cite{mayne2014model}, reference governors \cite{garone2017reference}, set invariance based approaches such as control barrier functions \cite{ames2016control}, barrier lyapunov functions \cite{tee2009barrier}, funnel control \cite{ilchmann2002tracking}, and prescribed performance control \cite{bechlioulis2008robust}.

Funnel-based control methods offer low-complexity and robust (model-free) control designs for handling (time-varying) output constraints for uncertain nonlinear systems. During the past years, funnel-based control designs were particularly utilized to ensure a user-defined transient and steady-state performance on output tracking/stabilization errors through confining the evolution of the output error signals within predefined time-varying funnels. Two main control approaches that have been proposed to handle the aforementioned objective are Funnel Control (FC) \cite{ilchmann2007tracking, lee2019asymptotic, berger2021funnel} and Prescribed Performance Control (PPC) \cite{bechlioulis2010prescribed,bechlioulis2014low,theodorakopoulos2015low}. Funnel control builds on the adaptive high-gain control methodology, where a time-varying and state-dependent function is replaced by the monotonically increasing control gain. In PPC,  initially, a transformation that incorporates the desired performance specifications is defined. Then, an appropriate control action is designed that establishes the uniform boundedness of the transformed system and gives the necessary and sufficient conditions for the satisfaction of the predefined performance (funnel) constraint. PPC was first presented in \cite{bechlioulis2008robust,bechlioulis2010prescribed} within a robust adaptive control framework for systems having known high relative degree and later was further extended to the approximation-free paradigm in \cite{bechlioulis2014low,theodorakopoulos2015low}, significantly reducing the complexity of the designed controller. Besides FC and PPC, Time-Varying Barrier Lyapunov Functions (TVBLFs) have been also employed to deal with similar problems, e.g., \cite{tee2011control}.

Despite their successful application for controlling output constrained systems, FC, PPC, and TVBLFs have been mainly focused on satisfying output performance constraints with specific cases of safety specifications (i.e., usually as constant upper and lower bounds on the tracking errors) provided that safety specifications are consistent (compatible) with the performance requirements. However, in general, performance constraints on tracking/stabilization may not be always in agreement with safety specifications. Recently, Control Barrier Functions (CBFs) have been introduced, for handling tracking/stabilization in the presence of safety specifications through the application of Quadratic Programs (QPs), in which safety specifications are considered as hard constraints and stabilization/tracking requirements are considered as a soft constraint\footnote{Note that under this scheme, the quality of the tracking/stabilization task is not predetermined (constrained), unlike the funnel-based control methods.}in the optimization problem \cite{ames2016control, ames2019control}. While typical CBFs are confined to constant constraints, recently \cite{xu2018constrained} presented control synthesis using time-varying CBFs to deal with time-varying hard output constraints. Nevertheless, CBF-based control synthesis requires exact knowledge of the system dynamics.

In the present paper, we propose a novel funnel-based control scheme capable of handling time-varying soft and hard (funnel-like) output constraints. First, we provide a novel online funnel planning scheme that constructs a Constraint Consistent Funnel (CCF) for each output of the system. Hard output constraints are always respected and soft output constraints are met only when they are not conflicting with the hard constraints. Then, a model-free robust controller is designed under a low-complexity PPC scheme to keep each system's output within its corresponding (online) planned CCF. The controller design is provided for uncertain nonlinear Euler-Lagrangian systems, which constitute a large class of practical physical systems (mobile robotic vehicles, robot manipulators, etc.). To the best of the authors' knowledge, this is the first work that considers hard and soft constraints under funnel-based control designs. Moreover, concerning the existing time-varying CBF-based control synthesis methods our work provides model-free and more computationally tractable control laws (i.e., optimization-free) to deal with time-varying hard/soft output constraints.

%%%%%%%%%%%%%%%%%%%%%%%%%%%%%%%%%%%%%%%%%%%%%%%%%%%%%%%%%%%%%%%%%%%%%%%%%%%%%%%%
\section{Problem Formulation} \label{sec:prob_formu}

Consider the following Euler-Lagrange (EL) system:
\begin{equation} \label{eq:sys_dynamics}
	\begin{cases}
		M(x)\dot{v} + C(x,v) v+ g(x) + D(x)v = u + d(t), \\
		y = x,
	\end{cases}
\end{equation}
where $x \coloneqq \col(x_i) \coloneqq [x_1, x_2, \ldots, x_n]^\top \in \R^{n}$ and $v \coloneqq \col(v_i) = \dot{x}$ are the generalized coordinates and their velocities, respectively, $M(\mathord{\cdot}): \R^n \rightarrow \R^{n \times n}$ is the inertia matrix, $C(\mathord{\cdot}): \R^n \times \R^n \rightarrow \R^{n \times n}$ is the centrifugal and Coriolis forces matrix, $g(\mathord{\cdot}): \R^n \rightarrow \R^n$ is the vector of gravitational forces, $D(\mathord{\cdot}): \R^n \rightarrow \R^{n \times n}$ is the matrix of friction-like terms, $d(\mathord{\cdot}): \R \rightarrow \R^n$ is the vector of unknown bounded piecewise continuous external disturbances, $u \in \R^n$ and $y \in \R^n$ denote the control input and the output, respectively. 
\begin{assumption} \label{assump:unknown}
	 $M(x), C(x,v), g(x), D(x)$, and the upper bound of $d(t)$ are unknown for the controller design.
\end{assumption}

The following property holds for the EL system \eqref{eq:sys_dynamics}:

\begin{property}\label{prop:conti}
	 $M(\mathord{\cdot})$ is symmetric and positive definite. Moreover, $M(x), g(x), D(x)$ are continuous over $x$ and $C(x,v)$ is continuous over $x$ and $v$.
\end{property}

Suppose that the outputs of \eqref{eq:sys_dynamics} are required to satisfy the following time-varying constraints:
\begin{align}
	\rhl_i(t)< x_i(t) < &\rhu_i(t), \quad i = \{1,\ldots, n\} \label{hard_const}, \\
	\rsl_i(t)< x_i(t) < &\rsu_i(t), \quad i = \{1,\ldots, n\} \label{soft_const},
\end{align}
where $\rhl_i(\mathord{\cdot}), \rhu_i(\mathord{\cdot}), \rsl_i(\mathord{\cdot}), \rsu_i(\mathord{\cdot}): \R \rightarrow \R,  i = \{1, \ldots, n\}$ are bounded continuously differentiable functions of time with bounded derivatives. Let $\rhl_i(t), \rhu_i(t)$ and $\rsl_i(t), \rsu_i(t)$ represent hard and soft constraints on $x_i(t)$, respectively.

\begin{assumption}[Feasibility of hard/soft constraints] \label{assum_soft_hard_bounds} 
Let	$\rhu_i(t) - \rhl_i(t) \geq \epsilon^h_i > 0$ and $\rsu_i(t) - \rsl_i(t) \geq \epsilon^s_i > 0$, $\forall t \geq 0, i = \{1, \ldots, n\}$.
\end{assumption} 

Note that under Assumption \ref{assum_soft_hard_bounds}, inequalities \eqref{hard_const} and \eqref{soft_const} denote separate hard and soft constrained feasible time-varying funnels for $x_i(t)$. Hard and soft constraints on $x_i(t)$  are said to be \textit{compatible} whenever both \eqref{hard_const} and \eqref{soft_const} can be satisfied at the same time (see Fig.\ref{fig:compa}). In other words, \eqref{hard_const} and \eqref{soft_const} on $x_i(t)$ are compatible at time $t$ if both $\rhu_i(t) > \rsl_i(t)$ and $\rsu_i(t) > \rhl_i(t)$ hold. 

For ease of presentation, we introduce the following assumption, which can be further relaxed (see the Appendix for the relaxed formulation).
\begin{assumption} \label{assum: ini_compat}
	 The hard \eqref{hard_const} and soft \eqref{soft_const} constraints are compatible at $t=0$ and the initial conditions, $x_i(0), i = \{1,\ldots, n\}$ satisfy both \eqref{hard_const} and \eqref{soft_const} at $t = 0$.
\end{assumption}

\textbf{Problem (Hard and soft constrained funnel control):} Given the aforementioned hard and soft output constraints in \eqref{hard_const} and \eqref{soft_const}, design under Assumptions 1-3:
 \begin{enumerate}
 	\item a continuous time-varying \textit{Constraint Consistent Funnel} (CCF) with boundary functions $\ru_i(t), \rl_i(t): \R_{\geq 0} \rightarrow \R$ for each output $x_i(t),i = \{1,\ldots, n\}$, where $\ru_i(t) - \rl_i(t) \geq \epsilon^c_i > 0, \forall t \geq 0$ (funnel feasibility condition);
 	\item a robust model-free control law $u(t,x)$ to ensure:
 \end{enumerate}
 	\begin{equation} \label{const_consist_funnel}
	\rl_i(t)< x_i(t) < \ru_i(t), \;\; \forall t\geq 0, \;\; i = \{1,\ldots, n\},\!\!\!\!\!\! 
	\vspace{-0.1cm}
\end{equation}
where constraint consistency of $\ru_i(t), \rl_i(t)$ means: (i) satisfaction of \eqref{const_consist_funnel} always implies satisfaction of \eqref{hard_const}, i.e., \eqref{const_consist_funnel} always respects the hard constraints \eqref{hard_const}, and (ii) whenever hard and soft constraints \eqref{hard_const} and \eqref{soft_const} are compatible, \eqref{const_consist_funnel} ensures satisfaction of \eqref{soft_const} (or \textit{exponentially fast} recovery of \eqref{soft_const}, which will be explained in detail later), i.e., \eqref{const_consist_funnel} respects (recovers) the soft constraints \eqref{soft_const} only when its satisfaction is not conflicting with the hard constraints \eqref{hard_const}. 

Fig.\ref{fig:const_consist_Reg} shows examples of \textit{Constraint Consistent Regions} (CCRs) for $x_i(t)$, in which hard constraints are always satisfied and the soft constraints are met only when they are compatible with the hard constraints. As can be seen in Fig.\ref{fig:compa}, if hard and soft constraints are compatible for all $t\geq0$, then the boundaries of the CCR can determine the boundary functions of a CCF in \eqref{const_consist_funnel}. However, this is not the case if hard and soft constraints become incompatible for a time interval since the upper or the lower boundary of the CCR becomes discontinuous and thus cannot be used to construct a feasible (well-defined) continuous CCF. For example, in Fig.\ref{fig:incompa} the upper bound of the CCR is discontinuous. Hence, to have a continuous transition region for the evolution of $x_i(t)$, $\ru_i(t)$ in \eqref{const_consist_funnel} needs to be designed (planned), as depicted in Fig.\ref{fig:incompa} (dashed curve), while a continuous $\rl_i(t)$ can be directly determined by the lower boundary of the CCR.

\begin{remark} \label{rem:app_examp}
	In practical applications, soft constraints \eqref{soft_const} can be considered as the required performance for reference tracking or stabilization, while hard constraints \eqref{hard_const} can be considered as safety requirements. For example, consider a mobile robot in the plane whose motion is modeled by \eqref{eq:sys_dynamics} and it requires to: (i) always remain in a box-shaped region indicated by $|x_i(t)| < s_{i}, i = \left\lbrace 1,2\right\rbrace $ for safety considerations, and (ii) track a desired time-varying continuously differentiable reference trajectory $x_d(t) = [x_{d_1}(t), x_{d_2}(t)]^\top$, such that $|x_i(t)- x_{d_i}(t)| < \gamma_i(t), i = \left\lbrace 1,2\right\rbrace $, where $\gamma_i(t)$ indicate user-defined positive (performance) functions decaying to a sufficiently small neighborhood of zero, e.g., $\gamma_i(t) = (\rho_{0_i} - \rho_{\infty_i}) \exp(-l_i t) + \rho_{\infty_i}$, in which $l_i, \rho_{\infty_i}>0$ determine the convergence rate and ultimate bound of the tracking errors, respectively, and $\rho_{0_i} > |x_i(0)- x_{d_i}(0)|, i = \left\lbrace 1,2\right\rbrace $. Note that, in general, the desired trajectory $x_d(t)$ may not always be within the safe region. In this application example, \eqref{hard_const} and \eqref{soft_const} will become $- s_i < x_i(t) < s_i$ and $ x_{d_i}(t) - \gamma_i(t)  < x_i(t) < x_{d_i}(t) + \gamma_i(t),i = \left\lbrace 1,2\right\rbrace $, respectively.
\end{remark} 
\begin{figure}[tbp]
%	\centering
		\flushleft
	\begin{subfigure}[t]{0.50\linewidth}
		%		\centering
		\flushleft
		\includegraphics[width=\linewidth]{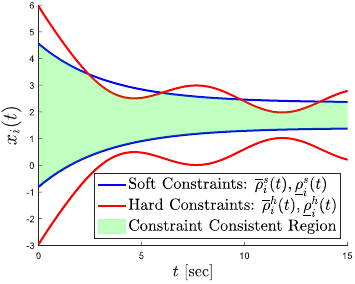}
		\caption{}
		\label{fig:compa}
	\end{subfigure}%
	~
	\begin{subfigure}[t]{0.5\linewidth}
		%		\centering
		\flushleft
		\includegraphics[width=\linewidth]{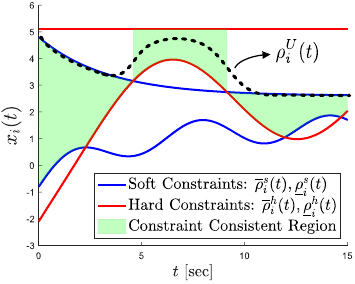}
		\caption{}
		\label{fig:incompa}
	\end{subfigure}
	\caption{(a) compatible ($\forall t \geq 0$), and (b) incompatible hard and soft constraints.\vspace{-0.4cm}}
	\label{fig:const_consist_Reg}
\end{figure}

%%%%%%%%%%%%%%%%%%%%%%%%%%%%%%%%%%%%%%%%%%%%%%%%%%%%%%%%%%%%%%%%%%%%%%%%%%%%%%%%
\section{Main Results}

In this section, given the hard and soft constraints in \eqref{hard_const} and \eqref{soft_const}, we will first propose an online funnel planning method to construct the constraint consistent feasible funnel boundary functions in \eqref{const_consist_funnel}. Then, we will design a robust model-free funnel-based control law using the prescribed performance control method to ensure \eqref{const_consist_funnel}.

%---------------------------------------------------------------------------------------
\subsection{Online Constraint Consistent Funnel Planning} \label{sec:funnel_plan}
Consider the hard and soft constraints in \eqref{hard_const} and \eqref{soft_const}. Note that whenever hard and soft constraints are compatible one can simply choose $\rl_i(t) = \max_t \{\rsl_i(t), \rhl_i(t)\}$ and $\ru_i(t) = \min_t \{\rsu_i(t), \rhu_i(t)\}, i = \{1,\ldots, n\}$ as proper candidates for CCF boundary functions within which $x_i(t)$ is allowed to evolve\footnote{By the notation $\max_t$ (resp. $\min_t$) we denote taking $\max$ (resp. $\min$) of their (time-varying) arguments with respect to time $t$.}(see Fig.\ref{fig:compa}). However, whenever hard and soft constraints become incompatible for some time interval (see Fig.\ref{fig:incompa}) the above choice leads to $\min_t \{\rsu_i(t), \rhu_i(t)\} \leq \max_t \{\rsl_i(t), \rhl_i(t)\}$, which gives an infeasible funnel in \eqref{const_consist_funnel}. In this respect, for each $x_i(t), i = \{1,\ldots, n\}$, we design $\rl_i(t)$ and $\ru_i(t)$ in \eqref{const_consist_funnel} as follows:
\begin{subequations} \label{non_smooth_funn_modif}
	\begin{numcases}{}
		\rl_i(t) \coloneqq \max_t \{\rsl_i(t) - \pl_i(t), \rhl_i(t)\}, \label{eq:funn_low_bound} \\
		\ru_i(t) \coloneqq \min_t \{\rsu_i(t) + \pu_i(t), \rhu_i(t)\}, \label{eq:funn_up_bound}
	\end{numcases}
\end{subequations}
where $\pl_i(t), \pu_i(t): \R_{\geq 0} \rightarrow \R_{\geq 0}, i = \{1,\ldots, n\}$ are continuous nonnegative \emph{modification signals} that are governed by the following dynamics:
\begin{subequations} \label{eq:modif_signals_dyn}
	\begin{numcases}{}
		\dpl_i =  \dfrac{1}{2} \big( 1 - \sign(\etl_i - \mu)\big) \dfrac{1}{\etl_i + \pl_i} - k_c\, \pl_i, \label{eq:modif_signals_low} \\
		\dpu_i =  \dfrac{1}{2} \big( 1 - \sign(\etu_i - \mu)\big) \dfrac{1}{\etu_i + \pu_i}  - k_c \, \pu_i, \label{eq:modif_signals_up}
	\end{numcases}
\end{subequations}
with $\pl_i(0) =  \pu_i(0) = 0$, in which $\etl_i(t) = \rhu_i(t) - \rsl_i(t)$, $\etu_i(t) = \rsu_i(t) - \rhl_i(t), i = \{1,\ldots, n\}$, and $\mu, k_c > 0$ are user-defined arbitrary positive constants. The modification signals $\pl_i(t), \pu_i(t)$ in \eqref{non_smooth_funn_modif}, adjust $\rl_i(t)$ and $\ru_i(t)$ whenever hard and soft constraints on $x_i(t)$ become conflicting so that the soft constraints \eqref{soft_const} are violated in favor of satisfying the hard constraints \eqref{hard_const}. 

In the sequel, we summarize the philosophy behind adopting \eqref{non_smooth_funn_modif} and \eqref{eq:modif_signals_dyn}, as well as the impact of  choosing $\mu$ and $k_c$. Recall that when hard and soft constraints are compatible and $\pl_i(t) = \pu_i(t) = 0$, \eqref{non_smooth_funn_modif} determines continuous boundary functions of the CCF \eqref{const_consist_funnel}. Therefore, due to Assumption \ref{assum: ini_compat} and  $\pl_i(0) = \pu_i(0) = 0$, \eqref{non_smooth_funn_modif} gives a feasible CCF for $x_i(t)$ at $t=0$. Note that, the hard and soft constraints \eqref{hard_const} and \eqref{soft_const} become conflicting (incompatible) when there exists $t>0$ such that: (i) $\rhu_i(t) - \rsl_i(t) < 0$ or (ii) $\rsu_i(t) - \rhl_i(t) < 0$, which lead to discontinuous CCR boundaries as mentioned in Section \ref{sec:prob_formu} (see Fig.\ref{fig:incompa}). Now consider case (ii) and let $\rsu_i(t) < \rhu_i(t), \forall t \geq 0$ (as it is illustrated in Fig.\ref{fig:incompa}). Based on a user-defined minimal distance $\mu>0$ between $\rhl_i(t)$ and $\rsu_i(t)$, the idea is to design a triggering process, by which $\ru_i(t)$ in \eqref{eq:funn_up_bound} starts disregarding the soft constraint $\rsu_i(t)$, thus allowing the output $x_i(t)$ to enter the region $\rsu_i(t)  < x_i(t) < \rsu_i(t) + \pu_i(t)$. In this respect, as soon as $\rsu_i(t) - \rhl_i(t) \leq \mu$, the term ${1}/(\etu_i + \pu_i) = {1}/((\rsu_i + \pu_i) -\rhl_i)$ in \eqref{eq:modif_signals_up} becomes active and sufficiently increases $\pu_i(t)$. Note that the magnitude of $\mu$ determines the level of conservatism for triggering the process of disregarding the soft constraints when hard and soft constraints tend to become conflicting. On the other hand, whenever the conflict between hard and soft constraints is resolved (i.e., $\rsu_i(t) - \rhl_i(t) > \mu$), \eqref{eq:modif_signals_up} reduces to $\dpu_i = -k_c \pu_i$ and ensures exponential convergence of $\pu_i(t)$ to zero. Owing to \eqref{eq:modif_signals_up}, this allows $\ru_i(t)$ to converge exponentially towards $\rsu_i(t)$ (i.e., the violated soft constraint gets \textit{recovered exponentially fast}). In this case, the rate of convergence of  $\ru_i(t)$ to the soft constraint $\rsu_i(t)$ can be adjusted by tuning $k_c$. Moreover, a larger $k_c$ can impede the growth of $\pu_i(t)$ (also $\pl_i(t)$), leading to a less conservative violation of the soft constraints. Finally, notice that if at some time $\pu_i(t)$ increases such that $\rsu_i(t) + \pu_i(t) >  \rhu_i(t)$, according to \eqref{eq:funn_up_bound}, $\ru_i(t)$ will become equal to $\rhu_i(t)$ to respect the hard constraint and due to Assumption \ref{assum_soft_hard_bounds} the CCF's boundaries $\rl_i(t) = \rhl_i(t)$ and $\ru_i(t) = \rhu_i(t)$ will remain feasible. In a similar fashion, one can justify the modification of $\rl_i(t)$ in \eqref{eq:funn_low_bound}. Finally, we emphasize that $\rl_i(t)$ and $\ru_i(t), i = \{1,\ldots, n\}$ obtained from \eqref{non_smooth_funn_modif} are continuous (but in general nonsmooth) functions of time. Moreover, from \eqref{eq:modif_signals_dyn}, $\dpl_i(t)$ and $\dpu_i(t)$ are piecewise continuous functions of time.

Before presenting the following lemma, we emphasize that in this work each potentially conflicting pair of hard and soft constraint bounds, namely, $\rhu_i(t)$ with $\rsl_i(t)$ and $\rhl_i(t)$ with $\rsu_i(t)$, is assumed to evolve in a well-posed manner. Specifically, there exists a dwell time between any two consecutive crossings of these bounds. Equivalently, for each $i \in \{1,\ldots,n\}$, the gap functions $\etl_i(t)=\rhu_i(t)-\rsl_i(t)$ and $\etu_i(t)=\rsu_i(t)-\rhl_i(t)$ may change sign only after a nonzero dwell time; in particular, $\etl_i(t)$ and $\etu_i(t)$ cannot switch sign infinitely often. This mild regularity property is typically satisfied in practical applications.

\begin{lemma} \label{lem:one}
	Under Assumptions \ref{assum_soft_hard_bounds} and \ref{assum: ini_compat}, equations \eqref{non_smooth_funn_modif} and \eqref{eq:modif_signals_dyn} construct
	$\rl_i(t)$ and $\ru_i(t)$ such that: (i) $\dot{\rho}^L_i, \dot{\rho}^U_i, \rl_i, \ru_i \in \mathcal{L}_\infty$ (are bounded signals), and (ii) $\ru_i(t) - \rl_i(t) \geq \epsilon^c_i > 0, \forall t\geq0, i = \{1,\ldots, n\}$.
\end{lemma}
\begin{proof}
	First, we establish $\pl_i, \dpl_i \in \Linf, i \in \{1,\ldots, n\}$. Consider $\dpl_i$ given by \eqref{eq:modif_signals_low}, which operates in two modes:
	
	\textit{\textbf{Mode I.}} When $\etl_i(t) > \mu$, \eqref{eq:modif_signals_low} reduces to $\dpl_i = -k_c \pl_i$, thus $\pl_i(t)$ becomes exponentially stable and $\pl_i, \dpl_i \in \Linf$.
	
	\textit{\textbf{Mode II.}} When $\etl_i(t) \leq \mu$, the first term on the right hand-side of \eqref{eq:modif_signals_low} is active. It holds that if $\etl_i(t) + \pl_i(t) \nrightarrow 0$ (does not converge to zero), then $\dpl_i \in \Linf$. Note that, by assumption we have $\rhu_i, \drhu_i, \rsl_i, \drsl_i \in \Linf$, which leads to $\etl_i, \detl_i \in \Linf$. Now, let $\etl_i(t) + \pl_i(t) \rightarrow 0$, which also indicates $\detl_i + \dpl_i < 0$ or $\detl_i + \dpl_i > 0$ depending on the sign of $\etl_i(t) + \pl_i(t)$ at Mode II's activation time. Owing to \eqref{eq:modif_signals_low}, $\etl_i(t) + \pl_i(t) \rightarrow 0$ leads to $\dpl_i \rightarrow +\infty$, which also requires $\detl_i \rightarrow -\infty$ or $\detl_i \rightarrow +\infty$ , however, this is a contradiction, since we had assumed $\detl_i \in \Linf$. Therefore, $\etl_i(t) + \pl_i(t) \nrightarrow 0$ and $\dpl_i \in \Linf$. Now let $\pl_i \rightarrow +\infty$ (resp. $\pl_i \rightarrow -\infty$), since $\etl_i \in \Linf$, from \eqref{eq:modif_signals_low} we get $\dpl_i(t) \rightarrow - \infty$ (resp. $\dpl_i(t) \rightarrow + \infty$), which contradicts the infinite growth of $\pl_i$, thus $\pl_i \in \Linf$.
	
	Next we prove that if $\pl_i(0) \geq 0$, indeed we will have $\etl_i(t) + \pl_i(t) \geq \epsilon^L_i > 0$ ($\epsilon^L_i$ is a positive constant) and $\pl_i(t)\geq 0, \forall t \geq 0, i \in \{1,\ldots, n\}$. Consider a sequence of switching times between Modes I and II in \eqref{eq:modif_signals_low}: $\{t_1,\ldots, t_j\}, j \in \mathbb{N}$, where $0< t_1 < \ldots < t_j$. Notice that, due to well-posedness of time-varying hard and soft constraint bounds, $\etl_i(t)$ and $\etu_i(t)$ do not change sign infinitely often, therefore, Zeno behavior is excluded as $\etl_i(t)$ and $\etu_i(t)$ in \eqref{eq:modif_signals_dyn} are continuous time-varying signals independent of $\pl_i$ and $\pu_i$. 
	
	Now consider two cases:
	
	\textit{\textbf{Case I.}} Suppose that at $t = 0$ \eqref{eq:modif_signals_low} starts evolving under Mode I (that means $\etl_i(0) > \mu$). In this case, since \eqref{eq:modif_signals_low} in Mode I is an exponentially stable system if $\pl_i(0) \geq 0$, then $\pl_i(t) \geq 0$ for all $t \leq t_1$. As soon as Mode II becomes active at time $t = t_1$, we have $\etl_i(t_1) = \mu > 0$ and $\pl_i(t_1) \geq 0$. Since it is proved that $\etl_i(t) + \pl_i(t) \nrightarrow 0$, for $t_1 \leq t \leq t_2$ and we also have $\etl_i(t_1) + \pl_i(t_1) > 0$, we can infer that there exists a $\epsilon^L_i>0$ such that $\etl_i(t) + \pl_i(t) \geq \epsilon^L_i > 0$, for $t_1 \leq t \leq t_2$.  Moreover, as  $\etl_i(t) + \pl_i(t) > 0$, for $t_1 \leq t \leq t_2$, if $\pl_i(t) = 0$ for $t \in [t_1,t_2]$, $\dpl_i(t)$ in \eqref{eq:modif_signals_low} (operating in Mode II) will remain positive and prevents $\pl_i(t)$ from getting negative, thus $\pl_i(t) \geq 0$, for $t_1 \leq t \leq t_2$. Followed by this, when \eqref{eq:modif_signals_low} switches to Mode I at $t = t_3$, again we will have $\pl_i(t_3) \geq 0$, and $\etl_i(t_3) > \mu > 0$ so the above results can be repeatedly extended for any $t\geq t_3$. Therefore, $\etl_i(t) + \pl_i(t) \geq \epsilon^L_i > 0$ and $\pl_i(t) \geq 0, \forall t\geq0$. 
	
	\textit{\textbf{Case II.}} Now this time suppose that at $t=0$ \eqref{eq:modif_signals_low} starts running under Mode II (that means $\etl_i(0) \leq \mu$). Note that due to Assumption \ref{assum: ini_compat} we have $\etl_i(0) > 0$. Therefore, similarly to Case I, we can prove that  $\etl_i(t) + \pl_i(t) \geq \epsilon^L_i > 0$ and $\pl_i(t) \geq 0, \forall t\geq0$.
	
	In a similar fashion, based on \eqref{eq:modif_signals_up} one can show that: (i) $\etu_i(t) + \pu_i(t) \nrightarrow 0$ and $\pu_i, \dpu_i \in \Linf$, (ii) $\etu_i(t) + \pu_i(t) \geq \epsilon^U_i > 0$, and (iii) the nonnegativity of $\pu_i(t)$ (i.e., $\pu_i(t) \geq 0$), $\forall t\geq0, i \in \{1,\ldots, n\}$. 
	
	Owing to \eqref{non_smooth_funn_modif}, we know that $\drl_i(t) \in  \{\drsl_i(t) - \dpl_i(t), \drhl_i(t)\}$ and $\dru_i(t) \in \{\drsu_i(t) + \dpu_i(t),$ $\drhu_i(t) \}$. Since $\drsl_i, \drhl_i, \drsu_i, \drhu_i \in \mathcal{L}_{\infty}$ and $\dpl_i, \dpu_i \in \Linf$ hold, we get $\drl_i, \dru_i \in  \Linf, \forall t\geq0, i \in \{1,\ldots, n\}$. Moreover, followed by $\rsl_i, \rhl_i, \rsu_i, \rhu_i \in \mathcal{L}_{\infty}$ and boundedness of $\pl_i, \pu_i$, \eqref{non_smooth_funn_modif} establishes $\rL_i, \rU_i$ $\in \Linf, \forall t\geq0, i \in \{1,\ldots, n\}$. 
	
	Finally, from \eqref{non_smooth_funn_modif} we get  $\rU_i(t) - \rL_i(t) \in \{\rhu_i - \rhl_i,\rhu_i- \rsl_i + \pl_i, \rsu_i - \rhl_i + \pu_i , \rsu_i - \rsl_i + \pu_i + \pl_i \}$, where due to Assumption \ref{assum_soft_hard_bounds} and the nonnegativity of $\pu_i(t), \pl_i(t)$, the first and the fourth elements are lower bounded by a positive constant. Moreover, the second and third elements of the above set are always lower bounded by a positive constant, owing to $\etl_i(t) + \pl_i(t) \geq \epsilon^L_i > 0$ and $\etu_i(t) + \pu_i(t) \geq \epsilon^U_i > 0$, respectively. Therefore, there exists a positive constant $\epsilon^c_i \coloneqq \min\{\epsilon^h_i, \epsilon^s_i, \epsilon^L_i, \epsilon^U_i \}$ such that  $\rU_i(t) - \rL_i(t) \geq \epsilon^c_i > 0, \forall t \geq 0, i \in \{1,\ldots, n\}$.
\end{proof}

\begin{remark}[Smooth CCF] \label{rem:smooth_ccFunel}
		The online funnel planning scheme given by \eqref{non_smooth_funn_modif} and \eqref{eq:modif_signals_dyn} provides a continuous but (in general) nonsmooth CCF, which can lead to continuous but nonsmooth control inputs under funnel-based control design methods. There are two sources for nonsmoothness in $\rl_i(t)$, $\ru_i(t)$: (i) the nonsmooth switch in \eqref{eq:modif_signals_dyn} and (ii) the nonsmooth $\max$ and $\min$ operators in \eqref{non_smooth_funn_modif}. To generate smooth $\rl_i(t)$, $\ru_i(t)$ we can use smooth switching mechanisms instead of the nonsmooth ones in \eqref{eq:modif_signals_dyn}. In this respect, $\sign(\etl_i-\mu)$ and $\sign(\etu_i-\mu)$ in \eqref{eq:modif_signals_dyn} can be replaced by $\tanh(\kappa \, (\etl_i-\mu))$ and $\tanh(\kappa \, (\etu_i - \mu))$, respectively, where a larger $\kappa > 0$ captures the $\sign(\mathord{\cdot})$ function behavior better. Moreover, instead of nonsmooth $\max$ and $\min$ operators in \eqref{non_smooth_funn_modif}, one can utilize their smooth \textit{over-} and \textit{under-approximations} as follows:
		\begin{subequations} \label{smooth_max_min_modif_fun}
			\begin{align}
				\rl_i(t) &\coloneqq \dfrac{1}{\nu} \ln \left[ e^{\big( \nu \, (\rsl_i(t) - \pl_i(t)) \big)} + e^ {\left( \nu \,\rhl_i(t) \right)} \right] \nonumber \\
				&>  \max_t \{\rsl_i(t) - \pl_i(t), \rhl_i(t)\}, \label{smooth_fun_modif_up} \\	
				\ru_i(t) &\coloneqq -\dfrac{1}{\nu} \ln \Big[ e^{\big( -\nu \, (\rsu_i(t) + \pu_i(t)) \big)} + e^{\left( -\nu \, \rhu_i(t) \right)} \Big] \nonumber \\
				&< \min_t \{\rsu_i(t) + \pu_i(t), \rhu_i(t)\}, \label{smooth_fun_modif_low}
			\end{align}
		\end{subequations}
		where a larger $\nu>0$ gives a closer approximation. Note that using a very small $\nu$ in \eqref{smooth_max_min_modif_fun} might lead to a very conservative inner-approximation of the original (feasible) CCF, thus jeopardizing its feasibility, while a very large $\nu$ might lead to instability in numerical calculations.
\end{remark}

%---------------------------------------------------------------------------------------
\subsection{Funnel-Based Controller Design}
Now we design a low-complexity model-free robust funnel controller using a Prescribed Performance Control (PPC) method similar to \cite{bechlioulis2014low} to ensure that the output signals $x_i(t), i = \{1,\ldots, n\}$ will always remain within the online planned (in general asymmetric) CCFs in \eqref{const_consist_funnel}. The EL system \eqref{eq:sys_dynamics} can be re-written in state-space form as follows:
\begin{equation}\label{eq:sys_dyn_state_form}
	\!\!\!\!\!\!\!\!\!\begin{cases}
		\!\!\dot{x} = v, \\ 
		\!\!\dot{v} = M^{-1} (x)\left( - C(x,v) v - g(x) - D(x)v + u + d(t) \right). \!\!\!\!\!\!\!\!\!\!\!\!\!\!
	\end{cases}
\end{equation}
The controller design is two-fold: \textbf{(I)} velocity-level control design, and \textbf{(II)} acceleration-level control design.

\textbf{Step I-a.}\footnote{If Assumption~\ref{assum: relaxed} is used in place of Assumption~\ref{assum: ini_compat}, then, owing to the properties of \eqref{phi_ini}, this step can be omitted. See the Appendix for further details.} Given the initial condition $x(0)$ and hard constraints \eqref{hard_const} at $t = 0$, determine  boundary functions $\rsl_i(t), \rsu_i(t)$ of the soft constraints in \eqref{soft_const} (i.e., user-defined performance constraints) according to the control application (e.g., see Remark \ref{rem:app_examp}) such that conditions in Assumption \ref{assum: ini_compat} are met.

\textbf{Step I-b.} Define the normalized (w.r.t. the asymmetric funnel given by \eqref{const_consist_funnel}) system outputs as $\xh(t,x) = \col(\xh_i(t,x_i)) \in \R^n$, where:
\begin{equation} \label{eq:normal_x_hat}
	\xh_i(t,x_i) \coloneqq \frac{x_i - \tfrac{1}{2} (\ru_i(t) + \rl_i(t))}{\tfrac{1}{2} (\ru_i(t) - \rl_i(t))}, \; i = \{1,\ldots, n\}, 
\end{equation}
in which $\xh_i \in (-1,1)$ when $x_i \in (\rl_i(t), \ru_i(t))$. Moreover, define control related signals $\xi_x \coloneqq \diag(\xi_{x_i}(t,\xh_i)) \in \R^{n \times n}$ and $\epix = \col(\epixi(\xh_{i})) \in \R^n$, where:
\begin{subequations}
	\begin{align}
		\xi_{x_i}(t,\xh_i) &\coloneqq \frac{4}{(\ru_i(t) - \rl_i(t)) (1-\xh_i^2)}, \label{eq:normal_x}\\
		\epixi(\xh_i) &= T(\xh_i) \coloneqq \ln \left( \frac{1+\xh_i}{1 - \xh_i} \right). \label{eq:mapped_x}
	\end{align}
\end{subequations}
Finally, design the desired reference velocity vector as:
\begin{equation} \label{eq:vel_level_control}
	v_d(t,\xh) \coloneqq - k_x \, \xi_x \, \epix,
\end{equation}
with  $k_x > 0$, where $v_d(t,\xh) = \col(v_{d_i}(t,\xh_i)) \in \R^n$.

\textbf{Step II-a.} Define the velocity errors vector $e_v \coloneqq \col(e_{v_i}) = v - v_d \in \R^n$. Now the objective is to design the acceleration level controller $u$ in \eqref{eq:sys_dyn_state_form} to compensate the velocity errors by enforcing a (optionally symmetric) exponentially narrowing funnel on $e_{v_i}(t)$ indicated by:
\begin{equation} \label{vel_perfo_funnel}
	-\gamma^v_{i}(t)  < e_{v_i}(t) <  \gamma^v_{i}(t), \; \;\, \forall t\geq 0, \; \; i = \{1,\ldots, n\},
\end{equation}
where $\gamma^v_{i} (\mathord{\cdot}): \R \rightarrow \R_{>0}, i = \{1,\ldots, n\}$ are continuously differentiable positive performance functions for the velocity errors, which decay to a small neighborhood of zero. A possible choice for $\gamma^v_{i}(t)$ is $(\rho_{0_i}^v - \rho_{\infty_i}^v) \exp(-l_i^v t) + \rho_{\infty_i}^v$, where $l^v_i, \rho^v_{\infty_i}$ are user-defined positive constants and $\rho^v_{0_i} > |e_{v_i}(0)|$ that ensures $e_{v_i}(0) \in (-\gamma^v_{i}(0), \gamma^v_{i}(0)), i = \{1,\ldots, n\}$.

\textbf{Step II-b.} Similarly to the first step, define the normalized (w.r.t. the symmetric funnel given by \eqref{vel_perfo_funnel}) velocity errors as $\evh(t,e_{v}) = \col(\evih(t,e_{v_i})) \in \R^n$, where:
\begin{equation} \label{eq:normal_e_v}
	\evih(t,e_{v_i}) \coloneqq \frac{e_{v_i}}{\gamma^v_{i}(t)}, \quad i = \{1,\ldots, n\}, 
\end{equation}
in which $\evih \in (-1,1)$ when $e_{v_i} \in (-\gamma^v_{i}(t), \gamma^v_i(t))$. In addition, define control related signals $\xi_{v} \coloneqq \diag(\xi_{v_i}(t,\evih)) \in \R^{n \times n}$ and $\epiv = \col(\epivi(\evih)) \in \R^n$, where:
\begin{subequations}
	\begin{align}
		\xi_{v_i}(t,\evih) &\coloneqq \frac{2}{\gamma^v_i(t) \, (1-\evih^2)}, \label{eq:normal_ev}\\
		\epivi(\evih) &= T(\evih) \coloneqq \ln \left( \frac{1+\evih}{1 - \evih} \right). \label{eq:mapped_ev}
	\end{align}
\end{subequations}
Finally, design the control input $u$ for \eqref{eq:sys_dyn_state_form} as:
\begin{equation} \label{acce_level_control}
	u(t,\evh) \coloneqq - k_v \, \xi_v \, \epiv,
\end{equation}
with  $k_v > 0$, where $u(t,\evh) = \col(u_i(t,\hat{e}_{v_i})) \in \R^n$.
\begin{remark}
	The prescribed performance control technique guarantees prescribed transient and steady-state performance specifications that are encapsulated by a time-varying funnel condition as in \eqref{const_consist_funnel} (or \eqref{vel_perfo_funnel}). The basic idea is to enforce the normalized state in \eqref{eq:normal_x_hat} (resp. normalized error in \eqref{eq:normal_e_v}) to remain strictly within the set $(-1,1)$. This is achieved through exploiting a smooth, strictly increasing, bijective nonlinear mapping function $T(\mathord{\cdot}): (-1,1) \rightarrow (-\infty, +\infty)$, as in \eqref{eq:vel_level_control} (resp. \eqref{eq:mapped_ev}), which transforms the normalized state (resp. error) from constrained space $(-1,1)$ to an unconstrained one over $\R$. In particular when $x_i(0) \in (\rl_i(0), \ru_i(0))$ (resp. $e_{v_i}(0) \in (-\gamma^v_i(0), \gamma^v_i(0))$), the mapped signal $\epixi(\xh_i)$ (resp.  $\epivi(\evih)$) is initially well-defined. In this case, it is not difficult to verify that simply maintaining boundedness of the mapped signal $\epixi(\xh_i)$ (resp.  $\epivi(\evih)$) for all $t\geq0$ is equivalent to ensuring $\xh_i(t) \in (-1,1)$ (resp. $\evih(t) \in (-1,1)$).
\end{remark} 

%---------------------------------------------------------------------------------------
\subsection{Stability Analysis}

The following theorem summarizes the main result of this paper, where it is proven that the proposed feedback control law \eqref{acce_level_control} is capable of maintaining $x_i(t),  i = \{1,\ldots, n\}$ within the online planned CCFs in Section \ref{sec:funnel_plan}, thus solving the robust hard and soft output constrained funnel control problem for uncertain EL systems.
\begin{theorem} \label{th:main_theorem}
	Consider the Euler-Lagrange system \eqref{eq:sys_dynamics} with hard and soft output constraints \eqref{hard_const} and \eqref{soft_const} under Assumptions \ref{assump:unknown}, \ref{assum_soft_hard_bounds} and \ref{assum: ini_compat}. Given $\rl_i(t), \ru_i(t), i = \{1,\ldots, n\}$ obtained from the constraint consistent online funnel planning scheme in \eqref{non_smooth_funn_modif} and \eqref{eq:modif_signals_dyn}, the feedback control law \eqref{acce_level_control} guarantees satisfaction of $\rl_i(t)< x_i(t) < \ru_i(t), \forall t\geq 0, i = \{1,\ldots, n\}$, as well as boundedness of all closed-loop signals.\footnote{The results of this theorem remain valid if Assumption~\ref{assum: ini_compat} is replaced by Assumption~\ref{assum: relaxed}. Moreover, in this case the proof remains unchanged. Refer to the Appendix for further details.}
\end{theorem}
\begin{proof}
The proof comprises three phases. First, we show that $\xh_i(t,x_i)$, and  $\evih(t,e_{v_i}), i = \{1,\ldots, n\}$ remain within $(-1, 1)$ for a specific time interval $[0, \taum)$ (i.e., the existence and uniqueness of maximal solutions). Next, we prove that the proposed control scheme guarantees, for all $t \in [0, \taum)$: (i) the boundedness of all closed loop signals as well as (ii) that $\xh_i(t,x_i)$, and  $\evih(t,e_{v_i})$ remain strictly in a compact subset of $(-1, 1)$, which leads to $\taum=\infty$ (i.e., forward completeness), thus finalizing the proof.
	
Differentiating $\xh_i$ \eqref{eq:normal_x_hat} and $\evih$ \eqref{eq:normal_e_v} yields:
	\begin{subequations}
		\begin{flalign}
			&\dot{\xh}_i = \frac{2}{\ru_i - \rl_i} \left[ v_i -\tfrac{1}{2} \left( (\dru_i + \drl_i) + \xh_i ( \dru_i - \drl_i) \right) \right]\!,\!\!\!\!\!\! & \\
			&\dot{\hat{e}}_{v_i} = (\gamma^v_{i})^{-1} \left[ \dot{e}_{v_i} - \evih \dot{\gamma}^v_{i} \right], \quad i = \{1,\ldots, n\}.&
		\end{flalign}
	\end{subequations}
	Define $\rho_m \coloneqq \diag(\ru_i - \rl_i) \in \R^{n\times n}, \rho_a \coloneqq \col(\ru_i + \rl_i) \in \R^{n}, \gamma^v \coloneqq \diag(\gamma^v_i) \in \R^{n\times n}$, then the above equations can be written in stacked form as:
	\begin{subequations} \label{eq:stacked_normaliz_dyn}
	 	\begin{alignat}{2}
	 		&\dot{\xh} &&= 2 (\rho_m(t))^{-1} \left( v - \tfrac{1}{2} \left( \dot{\rho}_a(t)  + \dot{\rho}_m(t) \, \xh \right) \right),\label{eq:stacked_normaliz_dyn_xh} \\
	 		&\dot{\hat{e}}_{v} &&= (\gamma^v(t))^{-1} \left( \dot{e}_{v} - \dot{\gamma}^v(t) \, \evh  \right). \label{eq:stacked_normaliz_dyn_evh}
	 	\end{alignat}
	 \end{subequations}
 	 Now employing \eqref{eq:sys_dyn_state_form}, as well as the fact that $v  = v_d + \gamma^v \evh$, $\dot{e}_v = \dot{v} - \dot{v}_d$, and substituting \eqref{eq:vel_level_control} and \eqref{acce_level_control} into \eqref{eq:stacked_normaliz_dyn} gives:
 		\begin{subequations}
 		\begin{alignat}{3}
 			&\dot{\xh} &&\coloneqq  h_{\xh}(t,\xh, \evh) \nonumber \\
 			&  &&=  2 (\rho_m(t))^{-1} \left[ v_d(t,\xh) + \gamv(t) \evh - \tfrac{1}{2} \left( \dot{\rho}_a(t)  + \dot{\rho}_m(t) \, \xh \right) \right],  \nonumber \\
 			&\dot{\hat{e}}_{v} &&\coloneqq h_{\evh}(t,\xh,\evh) = (\gamma^v(t))^{-1} \Big[ M^{-1}(x) \big(-C(x,v) v -g(x) \nonumber \\
 			& && - D(x)v   + u(t,\evh) + d(t) \big) -\dot{v}_d(t,\xh) -  \dot{\gamma}^v(t) \, \evh \Big]. \nonumber
 		\end{alignat}
 	\end{subequations}
	Note that in the above equations, based on \eqref{eq:normal_x_hat}, $x$ can be rewritten as $x(t,\xh)$. Moreover, from $v = v_d + \gamma^v \evh$, and \eqref{eq:vel_level_control} we have $v = v(t,\xh,\evh)$. 
	Thus, the closed-loop dynamical system of $[\xh , \hat{e}_{v} ]^{\top}$can be written in compact form as:
	\begin{equation} \label{normalized_dyn}
		\begin{bmatrix}
			\dot{\xh} \\
			\dot{\hat{e}}_{v}
		\end{bmatrix}
		= \begin{bmatrix}
			h_{\xh}(t,\xh, \evh) \\
			h_{\evh}(t,\xh,\evh)
		\end{bmatrix} \eqqcolon h(t,\xh,\evh).
	\end{equation}
Let us also define an open set $\Omega_h = \Omega_{\xh} \times \Omega_{\evh} \subset \R^{2n}$ with $\Omega_{\xh} = \Omega_{\evh} = \underset{n-\mathrm{times}}{\underbrace{(-1,1) \times \ldots \times (-1,1)}}$.

\textbf{\textit{Phase I.}} It is clear that the set $\Omega_h$ is nonempty and open. Followed by Assumption \ref{assum: ini_compat} and usage of \eqref{non_smooth_funn_modif},\footnote{Alternatively, followed by Assumption \ref{assum: relaxed} and usage of \eqref{non_smooth_funn_modif} and \eqref{phi_ini}.} $x_i(0), i = \{1,\ldots, n\}$ satisfy \eqref{const_consist_funnel} at $t=0$, which based on \eqref{eq:normal_x_hat} ensures $\xh(0) \in \Omega_{\xh}$. Moreover, followed by \eqref{eq:normal_e_v}, selecting $\rho^v_{0_i} > |e_{v_i}(0)|$ in $\gamma^v_i(t), i = \{1,\ldots, n\}$ guarantees $e_{v_i}(0)$ to satisfy \eqref{vel_perfo_funnel} at $t=0$, which leads to $\evh(0) \in \Omega_{\evh}$ based on \eqref{eq:normal_e_v}. Additionally, $h(t,\xh,\evh)$ is piecewise continuous on $t$ and locally Lipschitz on $[\xh , \hat{e}_{v} ]^{\top}$over the set $\Omega_h$, Therefore, the hypotheses of Theorem 54 in \cite[p.~476]{sontag1998mathematical} hold and the existence and uniqueness of a maximal solution $[\xh(t) , \hat{e}_{v}(t) ]^{\top}$of \eqref{normalized_dyn} for a time interval $[0, \taum)$ such that $[\xh(t) , \hat{e}_{v}(t) ]^{\top} \in \Omega_h, \forall t \in [0, \taum)$ is guaranteed. Accordingly, $\xh(t) \in \Omega_{\xh}$ and $\evh(t) \in \Omega_{\evh}, \forall t \in [0, \taum)$. As a result, we can further infer that  $x_i(t)$ and  $e_{v_i}(t), i = \{1,\ldots, n\}$ are bounded for all $t \in [0, \taum)$ as in \eqref{const_consist_funnel} and \eqref{vel_perfo_funnel}, respectively. 

\textbf{\textit{Phase II.}} Owing to $\xh(t) \in \Omega_{\xh}$ and $\evh(t) \in \Omega_{\evh}, \forall t \in [0, \taum)$, which reveals $\xh_i(t) \in (-1,1)$ and $\evih(t) \in (-1,1), \forall t \in [0, \taum)$, $\xi_{x_i}, \xi_{v_i}, i = \{1,\ldots, n\}$ in \eqref{eq:normal_x} and \eqref{eq:normal_ev} are lower bounded by a positive constant $\forall t \in [0, \taum)$. Moreover, the signals $\epixi$ and $\epivi$, defined in \eqref{eq:mapped_x} and \eqref{eq:mapped_ev}, are well-defined for all $t \in [0, \taum)$.

\textbf{\textit{Step$\,$1:}} Consider the following positive definite and radially unbounded Lyapunov function candidate: $V_1 = \tfrac{1}{2} \epix^{\top} \epix$. Differentiating $V_1$ with respect to time, substituting \eqref{eq:stacked_normaliz_dyn_xh}, \eqref{eq:sys_dyn_state_form}, \eqref{eq:vel_level_control}, and exploiting \eqref{eq:normal_x} and diagonality of $\xi_x$ gives:
\begin{align}
	\dot{V}_1 &= \epix^{\top} \xi_x \left[ -k_x \xi_x \epix + \gamv \evh - \tfrac{1}{2}(\dot{\rho}_a + \dot{\rho}_m \xh) \right] \nonumber \\
	&\leq -k_x \|\xi_x\|^2 \|\epix\|^2 + \|\xi_x\| \|\epix\| \|\Psi\|, \label{eq:lyap1_first}
\end{align}
where $\Psi = \gamv \evh - \tfrac{1}{2}(\dot{\rho}_a + \dot{\rho}_m \xh) \in \R^n$. From Phase I we have $\evh, \xh \in \Linf, \forall t \in [0, \taum)$. Moreover, Lemma \ref{lem:one} indicates $\dot{\rho}_a, \dot{\rho}_m \in \Linf, \forall t\geq0$, and $\gamv \in \Linf$ by construction. Hence, $\Psi \in \Linf, \forall t \in [0, \taum)$. Let $0<\theta_x<k_x$ and $\sigma_x \coloneqq k_x-\theta_x > 0$ be constants; thus adding and subtracting $\theta_x \|\xi_x\|^2 \|\epix\|^2$ to the right-hand side of \eqref{eq:lyap1_first} yields:
\begin{align}
	\dot{V}_1 &\leq - \sigma_x \|\xi_x\|^2 \|\epix\|^2 - \|\xi_x\| \|\epix\| \left( \theta_x \|\xi_x\| \|\epix\| - \|\Psi\| \right) \nonumber \\
	&= - \sigma_x \|\xi_x\|^2 \|\epix\|^2, \;\; \forall \|\epix\| \geq \frac{\|\Psi\|}{\theta_x \|\xi_x\|}, \;\; \forall t \in [0, \taum), \nonumber
\end{align}
which, indicates that $\epix$ is Uniformly Ultimately Bounded (UUB) \cite[Theorem 4.18]{khalil2002noninear}, thus there exists an ultimate bound $\epixbar \in \R_{>0}$ independent of $\taum$ such that $\|\epix\| \leq \epixbar$, $\forall t \in [0, \taum)$. Moreover, by taking the inverse logarithmic function of \eqref{eq:mapped_x} and exploiting $\epixbar$, we obtain:
\begin{equation} \label{eq:inv_T_xh}
	-1 < \tfrac{e^{-\epixbar} -1}{e^{-\epixbar} + 1} \eqqcolon \underline{b}_{\xh_i} \leq \xh_i(t) \leq \bar{b}_{\xh_i} \coloneqq \tfrac{e^{\epixbar} -1}{e^{\epixbar} + 1} < 1,
\end{equation}
for all $t \in [0, \taum), i = \{1,\ldots, n\}$. Thus, based on \eqref{eq:normal_x}, $\xi_x \in \Linf, \forall t \in [0, \taum)$, hence, the designed desired reference velocity vector $v_d$ in \eqref{eq:vel_level_control}, remains bounded for all $t \in [0, \taum)$. Moreover, invoking $v = v_d + \gamma^v \evh$, we also conclude that $v(t) \in \Linf, \forall t \in [0, \taum)$. Finally, by taking the time derivative of $v_d(t,\xh)$, substituting \eqref{eq:stacked_normaliz_dyn_xh} and utilizing \eqref{eq:inv_T_xh} and Lemma \ref{lem:one}, it is straightforward to deduce $\dot{v}_d \in \Linf$ for all $t \in [0, \taum)$ as well.

\textbf{\textit{Step$\,$2:}} Consider the following positive definite and radially unbounded Lyapunov function candidate: $V_2 = \tfrac{1}{2} \epiv^{\top} \epiv$. Differentiating $V_2$ with respect to time, substituting \eqref{eq:stacked_normaliz_dyn_xh}, \eqref{eq:sys_dyn_state_form}, \eqref{acce_level_control}, and exploiting \eqref{eq:stacked_normaliz_dyn_evh} and diagonality of $\xi_v$ yields:
\begin{align}
	\dot{V}_2 &= \epiv^{\top} \xi_v \big[- M^{-1} k_v \xi_v \epiv - M^{-1} \big( (C+D)(v_d + \gamv \evh) \nonumber \\
	& \quad + g - d \big) - \dot{v}_d - \dot{\gamma}^v \evh \big] \nonumber \\
	&\leq - k_v \lambda  \|\xi_v\|^2 \|\epiv\|^2 + \|\xi_v\| \|\epiv\| \|\Gamma\|, \label{eq:lyap2_first}
\end{align}
where $\lambda$ is the minimum eigenvalue of the positive definite matrix $M^{-1}$ and $\Gamma \coloneqq M^{-1} \big( (C+D)(v_d + \gamv \evh) + g - d \big) - \dot{v}_d - \dot{\gamma}^v \evh$. Note that we have already showed $x,v \in \Linf, \forall t \in [0, \taum)$ in Phase I and Step 1 of Phase II, thus, Property \ref{prop:conti} for EL system \eqref{eq:sys_dynamics} implies boundlessness of the EL system's matrices and vectors; that is $\|C(x,v)\|, \|g(x)\|, \|D(x)\|, \|M^{-1}(x)\| \in \Linf$ for all $t \in [0, \taum)$. As a result, owing to $x,v \in \Linf, \forall t \in [0, \taum)$ and Property \ref{prop:conti} of \eqref{eq:sys_dynamics}, boundedness of $d(t), \dot{\gamma}^v, \forall t \geq 0$, as well as boundedness of $\dot{v}_d, \evh, \forall t \in [0, \taum)$ established in Phase I and Phase II-Step 1, we have $\Gamma \in \Linf, \forall t \in [0, \taum)$. Now let $0<\theta_v<k_v \lambda$ and $\sigma_v \coloneqq k_v \lambda -\theta_v > 0$ be constants; thus adding and subtracting $\theta_v \|\xi_v\|^2 \|\epiv\|^2$ in the right-hand side of \eqref{eq:lyap2_first} yields:
\begin{align}
	\dot{V}_2 &\leq - \sigma_v \|\xi_v\|^2 \|\epiv\|^2 - \|\xi_v\| \|\epiv\| \left( \theta_v \|\xi_v\| \|\epiv\| - \|\Gamma\| \right) \nonumber \\
	&= - \sigma_v \|\xi_v\|^2 \|\epiv\|^2, \;\; \forall \|\epiv\| \geq \frac{\|\Gamma\|}{\theta_v \|\xi_v\|}, \;\; \forall t \in [0, \taum), \nonumber
\end{align}
which concludes that $\epiv$ is UBB \cite{khalil2002noninear}, hence, there exists an ultimate bound $\epivbar \in \R_{>0}$ independent of $\taum$ such that $\|\epiv\| \leq \epivbar$ for $\forall t \in [0, \taum)$. Similarly to Step 1, taking the inverse of \eqref{eq:mapped_ev}, and using $\epivbar$ leads to:
\begin{equation} \label{eq:inv_T_evh}
	-1 < \tfrac{e^{-\epivbar} -1}{e^{-\epivbar} + 1} \eqqcolon \underline{b}_{\evih}  \leq \evih(t) \leq \bar{b}_{\evih} \coloneqq \tfrac{e^{\epivbar} -1}{e^{\epivbar} + 1} < 1,
\end{equation}
for all $t \in [0, \taum), i = \{1,\ldots, n\}$.  From \eqref{eq:normal_ev}, we can deduce that $\xi_{v} \in \Linf, \forall t \in [0, \taum)$, thus the designed control input $u$ in \eqref{acce_level_control} is bounded for all $t \in [0, \taum)$.

\textbf{\textit{Phase III.}} Now we shall establish that $\taum = \infty$. In this direction, notice by \eqref{eq:inv_T_xh} and \eqref{eq:inv_T_evh} that $[\xh(t), \evh(t)]^{\top} \in \Omega_h^{\prime} \coloneqq \Omega_{\xh}^{\prime} \times \Omega_{\evh}^{\prime}$, where $\Omega_{\xh}^{\prime} \coloneqq [\underline{b}_{\xh_1}, \bar{b}_{\xh_1}] \times \ldots \times [\underline{b}_{\xh_n}, \bar{b}_{\xh_n}]$ and  $\Omega_{\evh}^{\prime} \coloneqq [\underline{b}_{\hat{e}_{v_1}}, \bar{b}_{\hat{e}_{v_1}}] \times \ldots \times [\underline{b}_{\hat{e}_{v_n}}, \bar{b}_{\hat{e}_{v_n}}]$ are nonempty and compact subsets of $\Omega_{\xh}$ and $\Omega_{\evh}$, respectively. Hence, assuming a finite $\taum < \infty$, since $\Omega_{h}^{\prime} \subset \Omega_{h}$, Proposition C.3.6 in \cite[p.~481]{sontag1998mathematical} dictates the existence of a time instant $t^{\prime} \in [0,\taum)$ such that $[\xh(t^{\prime}), \evh(t^{\prime})]^{\top} \notin \Omega_{h}^{\prime}$, which is a contradiction. Therefore, $\taum = \infty$. As a result, all closed-loop control signals remain bounded $\forall t \geq 0$, and moreover $[\xh(t), \evh(t)]^{\top} \in \Omega_{h}^{\prime} \subset \Omega_{h}, \forall t \geq 0$. Finally, recall from \eqref{eq:inv_T_xh} that $\xh_i(t) \in [\underline{b}_{\xh_i}, \bar{b}_{\xh_i}] \subset (-1,1), \forall t \geq 0,  i = \{1,\ldots, n\}$, thus invoking \eqref{eq:normal_x_hat}, it can be deduced that  $\rl_i(t)< x_i(t) < \ru_i(t)$ for all $t \geq 0, i = \{1,\ldots, n\}$.
\end{proof}

%%%%%%%%%%%%%%%%%%%%%%%%%%%%%%%%%%%%%%%%%%%%%%%%%%%%%%%%%%%%%%%%%%%%%%%%%%%%%%%%
\section{Simulation Results}
Consider a mobile robot operating on a 2-D plane with kinematics and dynamics expressed by (see Fig.\ref{fig:robot} left):
\begin{flalign} \label{eq:robot_kin_dyn}
	&\begin{cases}
		\dot{p}_c = S(\theta)\psi \\ 
		\bar{M} \dot{\psi} + \bar{D} \psi = \bar{u} + \bar{d}(t) 
	\end{cases}\!\!\!\!\!\!, ~
	\begin{aligned}
	 	S(\theta)\! =\!\! \left[ \!\!\!
	 	\begin{array}{ccc}
	 		\cos \theta &\!\! \sin \theta &\!\! 0 \\
	 		 0 &\!\! 0 &\!\! 1
	 	\end{array}\!\!\!\right]^{\!\!\top}\!\!\!,\!\!\!\!\!\!\!
	\end{aligned}&
\end{flalign}
where $p_c = [x_c,y_c,\theta]^{\top}$is the position and orientation of the body frame $\{C\}$ relative to reference frame $\{O\}$, $\psi = [v_T, \dot{\theta}]^{\top}$, $v_T, \dot{\theta}$ are the transnational speed along the direction of $\theta$ and the angular speed about the vertical axis passing through $C$, respectively. Moreover, $\bar{M} = \diag(m,I)$, $m, I$ are mass and moment of inertia of the robot about the vertical axis, respectively, $\bar{D} \in \R^{2 \times 2}$ is a constant damping matrix, and $\bar{d}(t)$ is the vector of bounded external disturbances. To avoid the nonholonomic constraints, one can transform \eqref{eq:robot_kin_dyn} w.r.t. the hand position $p = [x_c,y_c]^{\top} + L [\cos \theta, \sin \theta]^{\top}$(see Fig.\ref{fig:robot} left) and obtain an equivalent EL form as in \eqref{eq:sys_dynamics} that meets Property 1 (with $x = p$ being the output of \eqref{eq:sys_dynamics}). The relations between $M(x), C(x,v), D(x), u, d(t)$ in \eqref{eq:sys_dynamics} and $\bar{M}, \bar{D}, \bar{u}, \bar{d}(t)$ in \eqref{eq:robot_kin_dyn}, as well as the numerical values of the robot's dynamical parameters can be found in \cite{cai2014adaptive}. 

Now consider the scenario described in Remark \ref{rem:app_examp}, and let $x_d(t) = [-1.5 + 5.8\cos(0.24t + 1.5), 5.8\sin(0.24t + 1.5)]^{\top}$ be the trajectory of a moving object (reference trajectory). Let $\rhl_1 = -6.58$, $\rhu_1 = 6.58$ and $\rhl_2 = -4.63$, $\rhu_2 = 4.63$ represent the box shaped hard constraints \eqref{hard_const} on the robot's hand position $p = x = [x_1, x_2]^{\top}$. Moreover, let $x(0) = [-3.19, 1.70]^{\top}$, $\theta(0) = -0.33$, and $\psi(0) = [0.2,-0.1]^{\top}$. In addition, assume $\gamma_i(t), i = \left\lbrace 1,2\right\rbrace $ as the user-defined trajectory tracking performance function (given in Remark \ref{rem:app_examp}) with $l_1 = l_2 = 0.7$ , $\rho_{\infty_1} = \rho_{\infty_2} = 0.2$ and $\rho_{0_1}, \rho_{0_2}$ are selected such that $\rho_{0_i} > |x_i(0)- x_{d_i}(0)|, i = \left\lbrace 1,2\right\rbrace$. Under this assumption the soft constraints \eqref{soft_const} on $x_i(t)$ (accounting for the trajectory tracking performance) are given by $\rsl_i(t) = x_{d_i}(t) - \gamma_i(t)$ and $\rsu_i(t) = x_{d_i}(t) + \gamma_i(t), i = \left\lbrace 1,2\right\rbrace$. The external disturbances are considered as $\bar{d}(t) = [0.75 \sin(2t + \frac{\pi}{3}) + 1.5 \cos(3t + \frac{3\pi}{7}), 0.25 \cos(3t+\frac{\pi}{6}) + 0.75 \sin(5t - \frac{\pi}{3})]^{\top}$(note that the transformed disturbance vector $d(t)$ is bounded in \eqref{eq:sys_dynamics} for bounded $\bar{d}(t)$, see \cite{cai2014adaptive}). 

\begin{figure}[!tbp]
	%	\centering
	\flushleft
	\begin{subfigure}[t]{0.25\linewidth}
		%		\centering
		\flushleft
		\includegraphics[align=c,width=\linewidth]{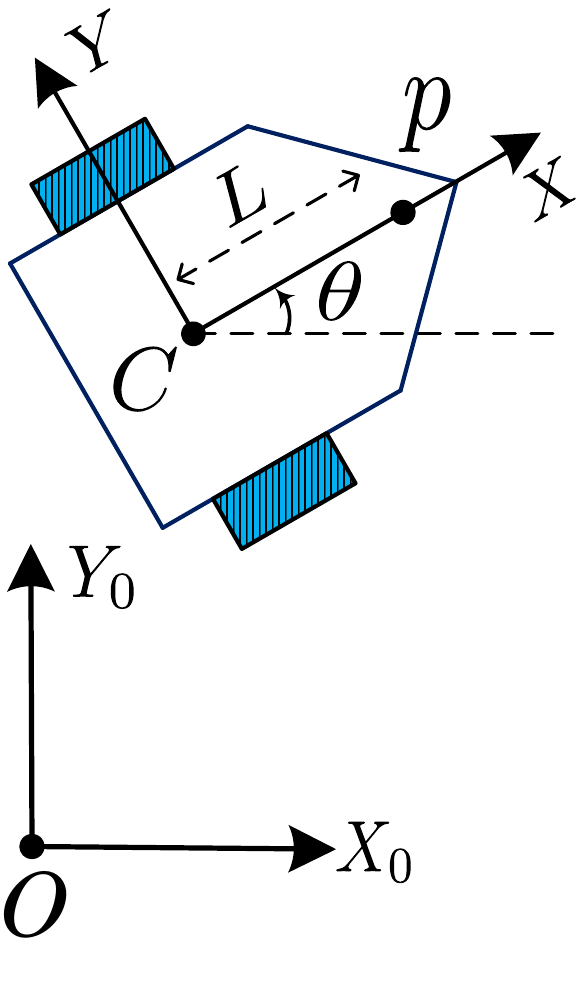}
		%		\caption{}
		%		\label{fig:robot}
	\end{subfigure}%
	~
	\begin{subfigure}[t]{0.75\linewidth}
		\centering
		%		\flushleft
		\includegraphics[align=c,width=\linewidth]{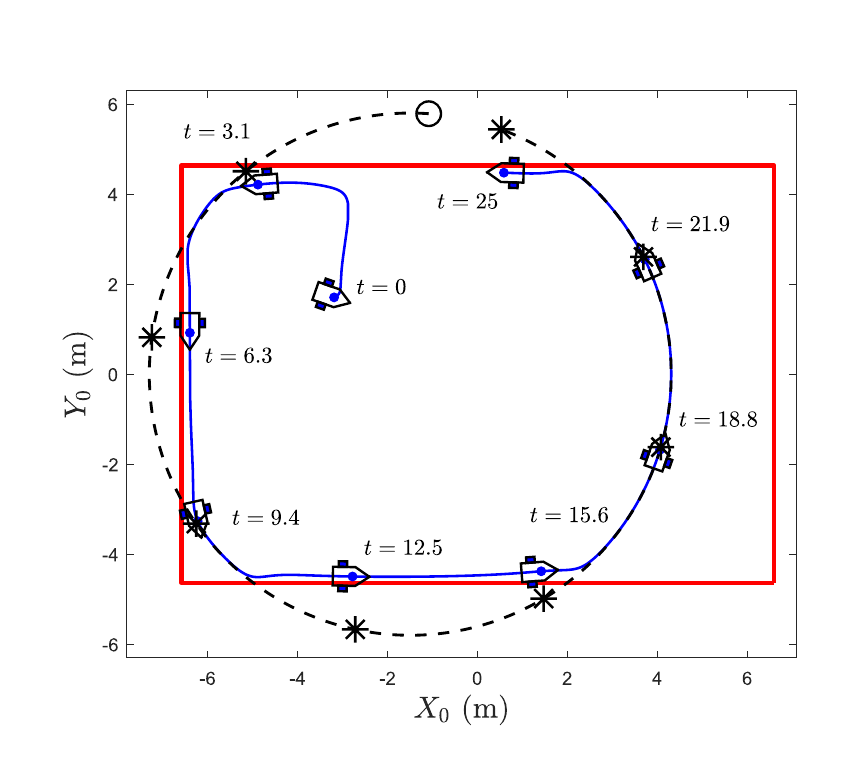}
		%		\caption{}
		\label{fig:robot_traj}
	\end{subfigure}%
	\caption{Left: mobile robot. Right: mobile robot's trajectory (blue line) tracking a moving object (dashed line) under hard constraints (red line) with $k_c = 3$ \vspace{-0.3cm}}
	\label{fig:robot}
\end{figure}

\begin{figure}[!tbp]
	\centering
	%	\flushleft
	\begin{subfigure}[t]{\linewidth}
		\centering
		%		\flushleft
		\includegraphics[width=\linewidth]{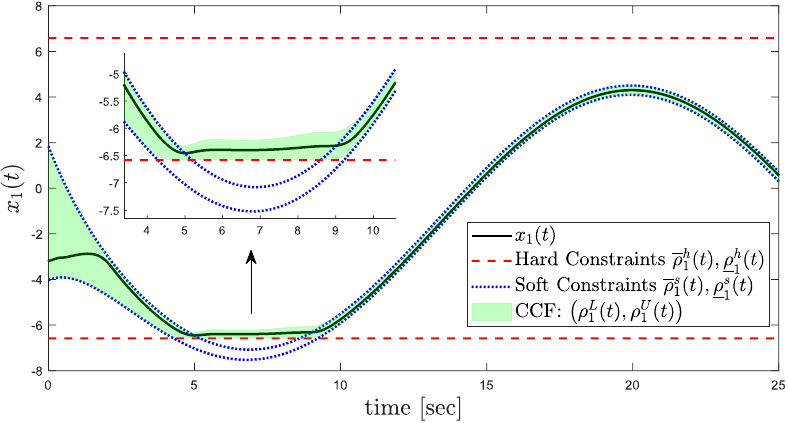}
		%		\caption{}
		%		\label{fig:com}
	\end{subfigure}%
	\vspace{0.1cm}
	%~
	\begin{subfigure}[t]{\linewidth}
		\centering
		%		\flushleft
		\includegraphics[width=\linewidth]{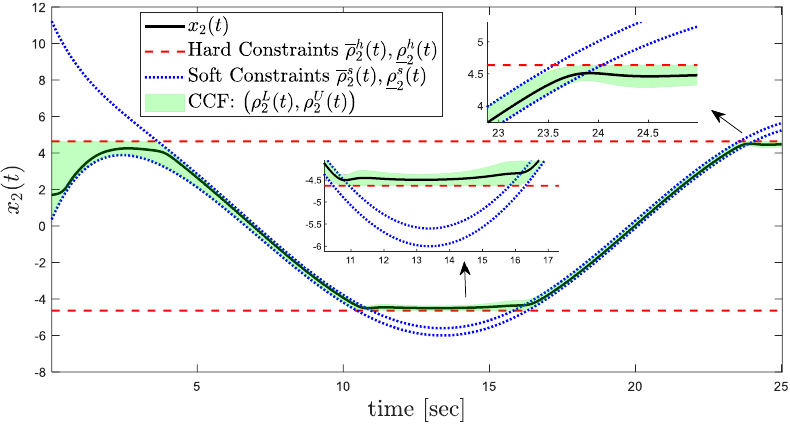}
		%		\caption{}
		%		\label{fig:}
	\end{subfigure}
	\caption{$x_1(t)$ and $x_2(t)$ evolution within CCFs under hard and soft constraints with $k_c = 3$.}
	\label{fig:x_1_x_2}
\end{figure}

\begin{figure}[!tbp]
	\centering
	%	\flushleft
	\begin{subfigure}[t]{0.5\linewidth}
		\centering
		%		\flushleft
		\includegraphics[width=\linewidth]{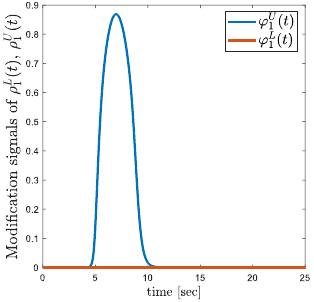}
		%		\caption{}
		%		\label{fig:com}
	\end{subfigure}%
	~
	\begin{subfigure}[t]{0.50\linewidth}
		\centering
		%		\flushleft
		\includegraphics[width=\linewidth]{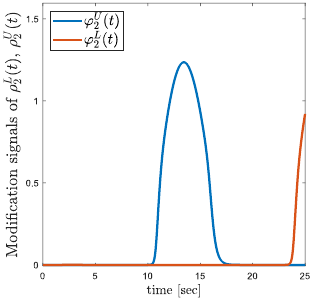}
		%		\caption{}
		%		\label{fig:}
	\end{subfigure}
	\caption{Evolution of the modification signals with $k_c = 3$.}
	\label{fig:phi_1_phi_2}
\end{figure}

The model-free control law \eqref{acce_level_control} designed for EL system \eqref{eq:sys_dynamics}, is applied to \eqref{eq:robot_kin_dyn} through the inverse transformation between $\bar{u}$ and  $u$ (see \cite{cai2014adaptive}). The parameters of $\gamma^v_{i}(t), i = \left\lbrace 1,2\right\rbrace$, in \eqref{vel_perfo_funnel}, employed in \eqref{acce_level_control} are considered as: $\rho_{\infty_1}^v = \rho_{\infty_2}^v = 0.1$, $l_1^v = l_2^v = 0.3$ and $\rho_{0_i}^v,  i = \left\lbrace 1,2\right\rbrace$ are selected such that $\rho^v_{0_i} > |e_{v_i}(0)|$. Moreover, $k_x = 0.2, k_v = 3$ are considered for \eqref{eq:vel_level_control} and \eqref{acce_level_control}, respectively. Finally, in the simulation we employed the smooth CCF planning proposed in Remark \ref{rem:smooth_ccFunel}, where $\mu = 0.01$, $k_c = 3$, $\kappa = 4$ and $\nu = 10$.

Under the proposed control scheme, Fig.\ref{fig:robot} on the right shows snapshots of the mobile robot's hand position trajectory (in solid blue) when tracking the moving object (depicted by $\ast$) whose trajectory and initial position are depicted by the dashed black line and $\Circle$, respectively. The red lines depict the hard (box-shaped) constraints on the robot's hand position, and it can be seen that the robot respects the hard constraints and tracks the object whenever it is possible. Fig.\ref{fig:x_1_x_2} depicts the evolution of the mobile robot's hand position in $X_0$ and $Y_0$ directions with time, i.e., $x_1(t)$ and $x_2(t)$, respectively. In Fig.\ref{fig:x_1_x_2}, the online planned constraint consistent funnels for each $x_1(t)$ and $x_2(t)$ are illustrated by the green regions, that satisfy the hard (safety) and  soft constraints (tracking performance) together. Finally, Fig.\ref{fig:phi_1_phi_2} shows the evolution of the nonnegative modification signals $\pl_i(t), \pu_i(t), i=\left\lbrace 1,2 \right\rbrace $ that contribute in generating the CCFs for $x_1(t)$ and $x_2(t)$.

Followed the discussion in Section \ref{sec:funnel_plan}, a smaller $k_c$ in \eqref{eq:modif_signals_dyn} leads to a slower  soft constraints recovery as well as more conservatism when the planned funnel violates the soft constraints. Figs. \ref{fig:robot_kc_small}, \ref{fig:x_1_x_2_kc_small} and \ref{fig:phi_1_phi_2_kc_small} show the simulation results with $k_c = 0.3$. As the figures suggest, in this case the mobile robot tends to violate the soft constraints more and have a slower rate for recovery of the soft constraints (notice that $\pu_i(t), \pl_i(t), i = \left\lbrace 1,2\right\rbrace$ show larger increases in this case).

%%%%%%%%%%%%%%%%%%%%%%%%%%%%%%%%%%%%%%%%%%%%%%%%%%%%%%%%%%%%%%%%%%%%%%%%%%%%%%%%
\section{Conclusions}
In this paper, we proposed a funnel control scheme under hard and soft time-varying output constraints for uncertain Euler Lagrange nonlinear systems, where hard and soft constraints resemble safety and performance specifications on the output, respectively. Given a set of hard and soft constraints, we proposed an online funnel planning scheme to design a constraint consistent funnel (CCF) that violates the soft constraints whenever they become conflicting with the hard constraints. The prescribed performance control approach was used to design a robust low complexity control law to maintain the outputs within the planned CCF. Future work will be devoted to generalizing the proposed scheme by considering time and output dependent hard/soft constraints and improving the funnel planning scheme to have exact (instead of exponential) soft constraint recovery.

\begin{figure}[!tbp]
	\centering
	%		\flushleft
	\includegraphics[width=\linewidth]{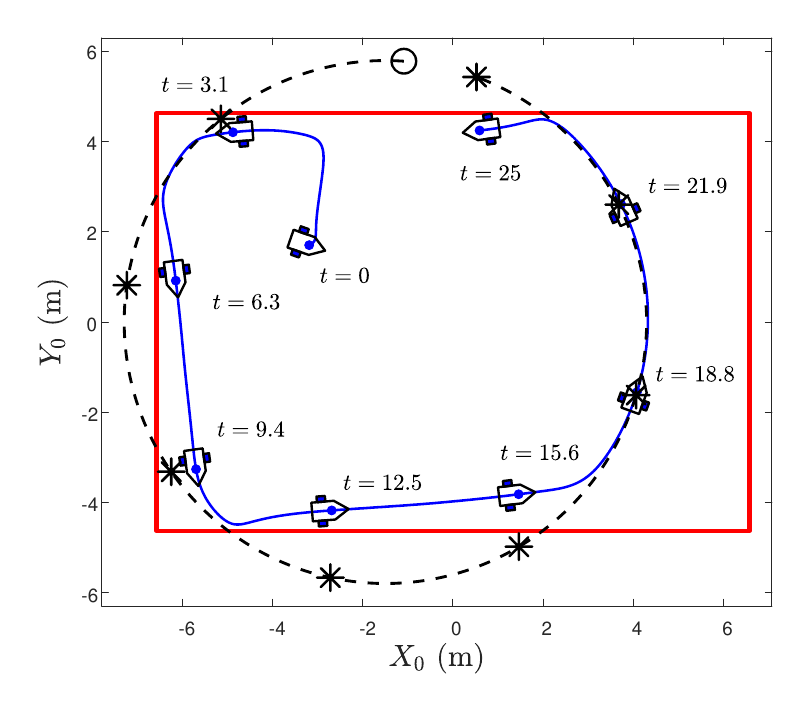}
	\caption{Mobile robot's trajectory (blue line) tracking a moving object (dashed line) under hard constraints (red line) with $k_c = 0.3$. \vspace{-0.3cm}}
	\label{fig:robot_kc_small}
\end{figure}

\begin{figure}[!tbp]
	\centering
	%	\flushleft
	\begin{subfigure}[t]{\linewidth}
		\centering
		%		\flushleft
		\includegraphics[width=\linewidth]{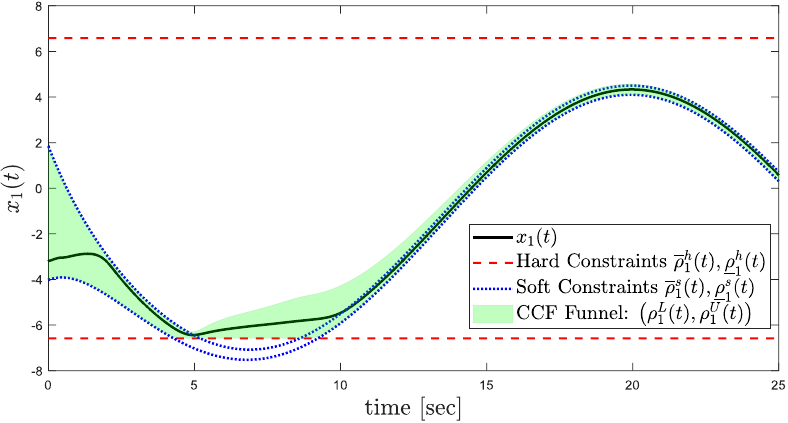}
		%		\caption{}
		%		\label{fig:com}
	\end{subfigure}%
	\vspace{0.1cm}
	%	~
	\begin{subfigure}[t]{\linewidth}
		\centering
		%		\flushleft
		\includegraphics[width=\linewidth]{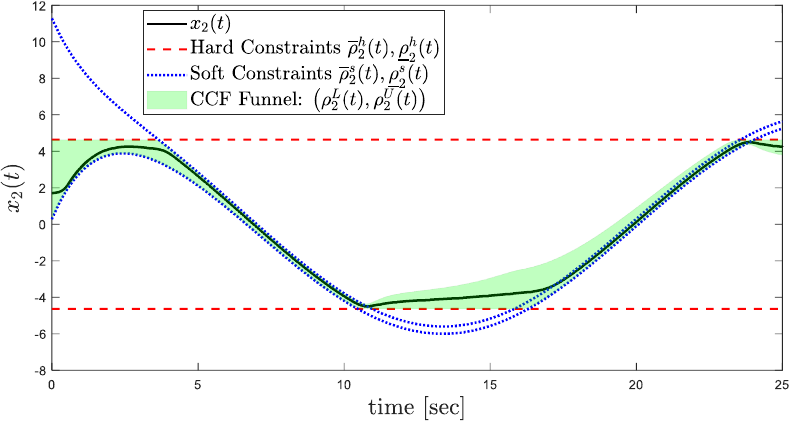}
		%		\caption{}
		%		\label{fig:}
	\end{subfigure}
	\caption{$x_1(t)$ and $x_2(t)$ evolution within CCFs under hard and soft constraints with $k_c = 0.3$ .}
	\label{fig:x_1_x_2_kc_small}
\end{figure}

\begin{figure}[!tbp]
	\centering
	%	\flushleft
	\begin{subfigure}[t]{0.50\linewidth}
		\centering
		%		\flushleft
		\includegraphics[width=\linewidth]{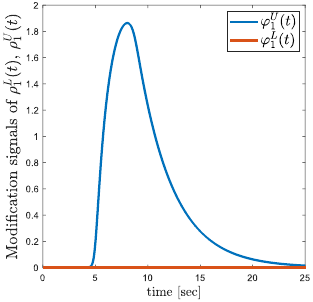}
		%		\caption{}
		%		\label{fig:com}
	\end{subfigure}%
	~
	\begin{subfigure}[t]{0.50\linewidth}
		\centering
		%		\flushleft
		\includegraphics[width=\linewidth]{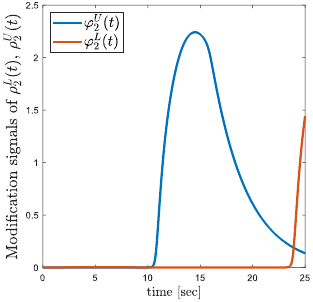}
		%		\caption{}
		%		\label{fig:}
	\end{subfigure}
	\caption{Evolution of the modification signals under $k_c = 0.3$. \vspace{-0.3cm}}
	\label{fig:phi_1_phi_2_kc_small}
\end{figure}

%%%%%%%%%%%%%%%%%%%%%%%%%%%%%%%%%%%%%%%%%%%%%%%%%%%%%%%%%%%%%%%%%%%%%%%%%%%%%%%%

\appendix[Relaxing Assumption \ref{assum: ini_compat}]
\label{appen:relax_assum}

Here, we explain how an appropriate choice of $\pl_i(0)$ and $\pu_i(0)$ in \eqref{eq:modif_signals_dyn} allows us to relax Assumption~\ref{assum: ini_compat}. Recall that Assumption~\ref{assum: ini_compat} requires, for each $x_i(t)$, that the hard and soft funnel constraints are compatible at $t = 0$ and that $x_i(0)$ satisfies both. In contrast, the following relaxed assumption only requires satisfaction of the hard funnel constraint at $t = 0$, which is a natural and non-restrictive condition.

\begin{assumption} \label{assum: relaxed}
	The initial outputs of \eqref{eq:sys_dynamics}, $x_i(0), i \in \{1,\ldots, n\}$, satisfy the hard funnel constraints \eqref{hard_const} at $t =0$, i.e., $x_i(0) \in (\rhl_i(0), \rhu_i(0))$, $i \in \{1,\ldots, n\}$.
\end{assumption}

To ensure the effectiveness of the proposed funnel-planning approach under Assumption~\ref{assum: relaxed}, we choose $\pl_i(0)$ and $\pu_i(0)$ in \eqref{eq:modif_signals_dyn} as follows:
\begin{subequations} \label{phi_ini}
	\begin{flalign}
		&\pl_i(0) =
		\begin{cases}
			\! 0 &  x_i(0)  - \rsl_i(0) > 0  \\
			\! \rsl_i(0) - x_i(0) + \varsigma^L_i \!& \text{otherwise}
		\end{cases}\!\!\!\!\!\!&  \label{phi_low_ini} \\
		&\pu_i(0) =
		\begin{cases}
			\! 0 & \rsu_i(0) - x_i(0) > 0  \\
			\! x_i(0) - \rsu_i(0) + \varsigma^U_i \!& \text{otherwise} \\
		\end{cases}\!\!\!\!\!\!&  \label{phi_up_ini}
	\end{flalign}
\end{subequations}
where $\varsigma^L_i$ and $\varsigma^U_i$ are some user-defined sufficiently small positive constants. 

First, note that \eqref{phi_ini} always guarantees $\pl_i(0), \pu_i(0) \geq 0$. Moreover, since $x_i(0) \in (\rhl_i(0), \rhu_i(0))$ holds by assumption, if $x_i(0) - \rsl_i(0) > 0$ and $\rsu_i(0) - x_i(0) > 0$ hold together it implies that the soft-constrained funnel \eqref{soft_const} is compatible with the hard-constrained funnel \eqref{hard_const} at $t = 0$. In this particular case, \eqref{phi_ini} yields $\pl_i(0) = \pu_i(0) = 0$.

Now consider the case where either $x_i(0) - \rsl_i(0) > 0$ or $\rsu_i(0) - x_i(0) > 0$ does not hold (note that, owing to $\rsu_i(0) > \rsl_i(0)$, only one of these two conditions can be violated at $t = 0$). In this case, $x_i(0)$ does not lie within the soft-constrained funnel (i.e., $x_i(0) \notin (\rsl_i(0), \rsu_i(0))$). Moreover, notice that in this scenario it is not determined whether the soft-constrained funnel \eqref{soft_const} is compatible with the hard-constrained funnel \eqref{hard_const} at $t = 0$. 

For example, consider \eqref{phi_up_ini} and let $\rsu_i(0) - x_i(0) \leq 0$. Employing \eqref{phi_up_ini} in \eqref{eq:funn_up_bound} yields $\rU_i(0) = \min \{ x_i(0) + \varsigma^U_i, \rhu_i(0)\}$, which provides a valid upper bound for $x_i(t)$ at $t = 0$ (i.e., $x_i(0) < \rU_i(0)$). Note that even if $\varsigma^U_i$ is chosen to be large enough such that $x_i(0) + \varsigma^U_i > \rhu_i(0)$, $\rU_i(0)$ becomes equal to $\rhu_i(0)$, which is still a valid upper bound according to Assumption \ref{assum: relaxed}.

Notice that, when using \eqref{phi_up_ini}, if the soft-constrained funnel is compatible with the hard-constrained funnel at $t = 0$ under the given user-defined minimal distance $\mu$ (i.e., $\etu_i = \rsu_i(t) - \rhl_i(t) > \mu$), then, as discussed in Section \ref{sec:funnel_plan}, \eqref{eq:modif_signals_up} ensures that $\pu_i(t)$ decreases exponentially towards zero. Consequently, by \eqref{eq:funn_up_bound}, the initially compatible soft constraint boundary $\rsu_i(t)$ is recovered exponentially fast. Conversely, if the soft-constrained funnel is initially incompatible with the hard-constrained funnel under the given user-defined minimal distance (i.e., $\etu_i = \rsu_i(t) - \rhl_i(t) \leq \mu$), then \eqref{eq:modif_signals_up} guarantees the evolution of $\pu_i(t)$ such that the online planned CCF boundary $\rU_i(t)$ approaches the incompatible soft constraint bound $\rsu_i(t)$ as closely as possible (depending on the tuning parameter $k_c$), while still respecting the hard-constrained funnel. 

Similarly, one can justify the use of \eqref{phi_low_ini} in choosing $\pl_i(0)$ under Assumption~\ref{assum: relaxed}.

Finally, we emphasize that if $\pl_i(0)$ and $\pu_i(0)$ are selected according to \eqref{phi_ini}, then it is straightforward to verify that the results of Lemma~\ref{lem:one} remain valid when Assumption~\ref{assum: ini_compat} is replaced by Assumption~\ref{assum: relaxed}. Consequently, the results of Theorem~\ref{th:main_theorem} also hold under (relaxed) Assumption~\ref{assum: relaxed}.

%%%%%%%%%%%%%%%%%%%%%%%%%%%%%%%%%%%%%%%%%%%%%%%%%%%%%%%%%%%%%%%%%%%%%%%%%%%%%%%%
\bibliographystyle{ieeetr}
\bibliography{Refs}

\begin{thebibliography}{10}

\bibitem{mayne2014model}
D.~Q. Mayne, ``Model predictive control: Recent developments and future
  promise,'' {\em Automatica}, vol.~50, no.~12, pp.~2967--2986, 2014.

\bibitem{garone2017reference}
E.~Garone, S.~Di~Cairano, and I.~Kolmanovsky, ``Reference and command governors
  for systems with constraints: A survey on theory and applications,'' {\em
  Automatica}, vol.~75, pp.~306--328, 2017.

\bibitem{ames2016control}
A.~D. Ames, X.~Xu, J.~W. Grizzle, and P.~Tabuada, ``Control barrier function
  based quadratic programs for safety critical systems,'' {\em IEEE
  Transactions on Automatic Control}, vol.~62, no.~8, pp.~3861--3876, 2016.

\bibitem{tee2009barrier}
K.~P. Tee, S.~S. Ge, and E.~H. Tay, ``Barrier lyapunov functions for the
  control of output-constrained nonlinear systems,'' {\em Automatica}, vol.~45,
  no.~4, pp.~918--927, 2009.

\bibitem{ilchmann2002tracking}
A.~Ilchmann, E.~P. Ryan, and C.~J. Sangwin, ``Tracking with prescribed
  transient behaviour,'' {\em ESAIM: Control, Optimisation and Calculus of
  Variations}, vol.~7, pp.~471--493, 2002.

\bibitem{bechlioulis2008robust}
C.~P. Bechlioulis and G.~A. Rovithakis, ``Robust adaptive control of feedback
  linearizable mimo nonlinear systems with prescribed performance,'' {\em IEEE
  Transactions on Automatic Control}, vol.~53, no.~9, pp.~2090--2099, 2008.

\bibitem{ilchmann2007tracking}
A.~Ilchmann, E.~P. Ryan, and P.~Townsend, ``Tracking with prescribed transient
  behavior for nonlinear systems of known relative degree,'' {\em SIAM Journal
  on Control and Optimization}, vol.~46, no.~1, pp.~210--230, 2007.

\bibitem{lee2019asymptotic}
J.~G. Lee and S.~Trenn, ``Asymptotic tracking via funnel control,'' in {\em
  2019 IEEE 58th Conference on Decision and Control (CDC)}, pp.~4228--4233,
  IEEE, 2019.

\bibitem{berger2021funnel}
T.~Berger, A.~Ilchmann, and E.~P. Ryan, ``Funnel control of nonlinear
  systems,'' {\em Mathematics of Control, Signals, and Systems}, vol.~33,
  no.~1, pp.~151--194, 2021.

\bibitem{bechlioulis2010prescribed}
C.~P. Bechlioulis and G.~A. Rovithakis, ``Prescribed performance adaptive
  control for multi-input multi-output affine in the control nonlinear
  systems,'' {\em IEEE Transactions on automatic control}, vol.~55, no.~5,
  pp.~1220--1226, 2010.

\bibitem{bechlioulis2014low}
C.~P. Bechlioulis and G.~A. Rovithakis, ``A low-complexity global
  approximation-free control scheme with prescribed performance for unknown
  pure feedback systems,'' {\em Automatica}, vol.~50, no.~4, pp.~1217--1226,
  2014.

\bibitem{theodorakopoulos2015low}
A.~Theodorakopoulos and G.~A. Rovithakis, ``Low-complexity prescribed
  performance control of uncertain mimo feedback linearizable systems,'' {\em
  IEEE Transactions on Automatic Control}, vol.~61, no.~7, pp.~1946--1952,
  2015.

\bibitem{tee2011control}
K.~P. Tee, B.~Ren, and S.~S. Ge, ``Control of nonlinear systems with
  time-varying output constraints,'' {\em Automatica}, vol.~47, no.~11,
  pp.~2511--2516, 2011.

\bibitem{ames2019control}
A.~D. Ames, S.~Coogan, M.~Egerstedt, G.~Notomista, K.~Sreenath, and P.~Tabuada,
  ``Control barrier functions: Theory and applications,'' in {\em 2019 18th
  European control conference (ECC)}, pp.~3420--3431, IEEE, 2019.

\bibitem{xu2018constrained}
X.~Xu, ``Constrained control of input--output linearizable systems using
  control sharing barrier functions,'' {\em Automatica}, vol.~87, pp.~195--201,
  2018.

\bibitem{sontag1998mathematical}
E.~D. Sontag, {\em Mathematical control theory: deterministic finite
  dimensional systems}.
\newblock Springer-Verlag New York, Inc., 1998.

\bibitem{khalil2002noninear}
H.~K. Khalil, {\em Noninear Systems}.
\newblock Prentice-Hall, New Jersey, 3~ed., 2002.

\bibitem{cai2014adaptive}
X.~Cai and M.~De~Queiroz, ``Adaptive rigidity-based formation control for
  multirobotic vehicles with dynamics,'' {\em IEEE Transactions on Control
  Systems Technology}, vol.~23, no.~1, pp.~389--396, 2014.

\end{thebibliography}

\end{document}